\documentclass[a4paper,USenglish,cleveref,autoref,thm-restate]{lipics-v2021}

\title{Recognizing 2-Layer and Outer $k$-Planar Graphs}

\author{Yasuaki Kobayashi}{Hokkaido University, Sapporo, Japan}{koba@ist.hokudai.ac.jp}{https://orcid.org/0000-0003-3244-6915}{Supported by JSPS KAKENHI Grant Numbers JP23K28034, JP24H00686, and JP24H00697.}

\author{Yuto Okada}{Nagoya University, Japan \and \url{https://yutookada.com/en}}{research@yutookada.com}{https://orcid.org/0000-0002-1156-0383}{Supported by JST SPRING, Grant Number JPMJSP2125 and JSPS KAKENHI Grant Number JP22H00513 (Hirotaka Ono).}

\author{Alexander Wolff}{Universit\"at W\"urzburg, Germany \and \url{https://www.informatik.uni-wuerzburg.de/en/algo/team/wolff-alexander}}{}{https://orcid.org/0000-0001-5872-718X}{}

\authorrunning{Y.~Kobayashi, Y.~Okada, and A.~Wolff}

\Copyright{Yasuaki Kobayashi, Yuto Okada, Alexander Wolff}



\acknowledgements{We thank Hirotaka Ono and the AFSA project (Creation and Organization of Innovative Algorithmic Foundations for Social Advancement) for supporting our work. We thank Boris Klemz and Marie Diana Sieper for useful discussions.}

\hideLIPIcs

\nolinenumbers %

\usepackage[disableredefinitions]{complexity}
\usepackage{mathtools}
\usepackage[textsize=small]{todonotes}
\usepackage{here}
\usepackage{amssymb}
\usepackage{cite}
\usepackage{tcolorbox}
\usepackage[linesnumbered,ruled,vlined]{algorithm2e}
\SetArgSty{}
\DontPrintSemicolon
\graphicspath{{./figures/}}%

\EventEditors{Oswin Aichholzer and Haitao Wang}
\EventNoEds{2}
\EventLongTitle{41st International Symposium on Computational Geometry (SoCG 2025)}
\EventShortTitle{SoCG 2025}
\EventAcronym{SoCG}
\EventYear{2025}
\EventDate{June 23--27, 2025}
\EventLocation{Kanazawa, Japan}
\EventLogo{socg-logo.pdf}
\SeriesVolume{332}
\ArticleNo{63}

\newcommand{\defproblem}[3]{
  \begin{tcolorbox}%
    \nolinenumbers\vspace*{-1ex}\hspace*{-2ex}
    \begin{minipage}{0.98\textwidth}
      \begin{tabular}{@{}>{\normalsize}l@{~~}>{\normalsize}p{0.92\textwidth}@{}}
        {\sf\bfseries\color{lipicsGray} Problem:} & #1\\[.1ex]
        {\sf\bfseries\color{lipicsGray} Input:} & #2\\[.1ex]
        {\sf\bfseries\color{lipicsGray} Question:} & #3
      \end{tabular}
    \end{minipage}\vspace*{-1ex}
  \end{tcolorbox}
}

\definecolor{defblue}{rgb}{0.121,0.47,0.705}
\definecolor{linkblue}{rgb}{0.098,0.098,0.4392}
\let\emph\relax
\DeclareTextFontCommand{\emph}{\color{defblue}\em}

\newcommand{\YES}{\rm{YES}\xspace}

\newcommand{\vtrue}{\ensuremath{\mathtt{true}}\xspace}
\newcommand{\vfalse}{\ensuremath{\mathtt{false}}\xspace}
\newcommand{\Dp}{\mathtt{draw}}

\newcommand{\XNLP}{{\sf XNLP}\xspace}
\newcommand{\XALP}{{\sf XALP}\xspace}

\newcommand{\DDDD}{\ensuremath{D_{\le i}}\xspace}
\newcommand{\DDD}{\ensuremath{D_{\le i-1}}\xspace}
\newcommand{\DD}{\ensuremath{D_{i-1}}\xspace}
\newcommand{\xx}{\ensuremath{\chi_{i-1}}\xspace}

\newcommand{\OuterkPlanarity}{\textsc{Outer $k$-Planarity}\xspace}
\newcommand{\TwoSidedkPlanarity}{\textsc{Two-Sided $k$-Planarity}\xspace}
\newcommand{\OneSidedkPlanarity}{\textsc{One-Sided $k$-Planarity}\xspace}

\newcommand{\Bandwidth}{\textsc{Bandwidth}\xspace}

\DeclareMathOperator{\bw}{bw}
\DeclareMathOperator{\cross}{cr}
\DeclareMathOperator{\Sum}{sum}

\begin{document}

\maketitle

\begin{abstract}
    The \emph{crossing number} of a graph is the least number of
    crossings over all drawings of the graph in the plane.
    Computing the crossing number of a given graph is \NP-hard, but
    fixed-parameter tractable (\FPT) with respect to the natural parameter.
    Two well-known variants of the problem are
    \emph{2-layer crossing minimization} and \emph{circular
        crossing minimization},
    where every vertex must lie on one of two \emph{layers},
    namely two parallel lines, or a circle, respectively.
    In both cases, edges are drawn as straight-line segments.
    Both variants are \NP-hard,
    but admit \FPT-algorithms with respect to the natural parameter.

    In recent years, in the context of beyond-planar graphs, a local
    version of the crossing number has also received considerable attention.
    A graph is \emph{$k$-planar} if it admits a drawing with at most
    $k$ crossings per edge.
    In contrast to the crossing number, recognizing $k$-planar graphs is \NP-hard
    even if $k=1$ and hence not likely to be \FPT\
    with respect to the natural parameter~$k$.

    In this paper, we consider the two above variants in the local setting.
    The $k$-planar graphs that admit a straight-line drawing with vertices on
    two layers or on a circle are called \emph{2-layer $k$-planar} and
    \emph{outer $k$-planar} graphs, respectively.
    We study the parameterized complexity of the two recognition
    problems with respect to the natural parameter~$k$.
    For $k=0$, the two classes of graphs are exactly the caterpillars
    and outerplanar graphs, respectively,
    which can be recognized in linear time.
    Two groups of researchers independently showed that outer 1-planar graphs can
    also be recognized in linear time [Hong et al., Algorithmica 2015; Auer et al., Algorithmica 2016].
    One group asked explicitly whether outer 2-planar graphs
    can be recognized in polynomial time.

    Our main contribution consists of \XP-algorithms for recognizing
    2-layer $k$-planar graphs and outer $k$-planar graphs, which implies
    that both recognition problems can be solved in polynomial time
    for every fixed~$k$.
    We complement these results by showing that recognizing 2-layer
    $k$-planar graphs is \XNLP-complete and that recognizing outer
    $k$-planar graphs is \XNLP-hard.
    This implies that both problems are \W$[t]$-hard for every $t$ and
    that it is unlikely that they admit \FPT-algorithms.
    On the other hand, we present an \FPT-algorithm for recognizing
    2-layer $k$-planar graphs where the order of the vertices on one
    layer is specified.
\end{abstract}

\section{Introduction}\label{sec:introduction}

When evaluating the quality of a graph drawing, one of the established
metrics is the number of crossings, whose importance is supported by user
experiments~\cite{purchase-etal-aesthetics-TVCG12}.  Unfortunately,
computing the crossing number of a given graph, that is, the minimum
number of crossings over all drawings of the graph, is
NP-hard~\cite{garey-johnson-cr-NP-SIJADM83}, even for graphs that
become planar after removal of a single
edge~\cite{cabello-mohar-near-planar-NP-hard-SICOMP13}.
On the other hand, the problem is fixed-parameter tractable (\FPT)
with respect to the natural parameter, that is, the number of
crossings~\cite{grohe-cr-FPT-JCSS04,kawarabayashi-reed-cr-FPT-STOC07}.
Many variants of the crossing number have been studied; see Schaefer's
survey~\cite{schaefer-survey-EJC24}.  Two variants with geometric
restrictions have attracted considerable attention: \emph{2-layer crossing
    minimization} and \emph{circular} (or \emph{convex}, or \emph{1-page})
\emph{crossing minimization}, where the placement of the vertices is
restricted to two parallel lines (called \emph{layers}) and to a
circle, respectively.  In both cases, edges are drawn as
straight-line segments.  Circular crossing minimization is \NP-hard,
but admits \FPT-algorithms with respect to the natural
parameter~\cite{bannister-eppstein-JGAA18,KobayashiOT17:IPEC:one-page-cm}.
In practice, often the so-called \emph{sifting heuristic} is
used~\cite{BaurB04:WG:circular-crossing-minimization}.
Circular crossing minimization can be seen as a special case of a
\emph{book embedding problem}, where vertices must lie on a straight
line, the \emph{spine} of the book, and each edge must be drawn on
one of a given number of halfplanes called \emph{pages} whose
intersection is the spine.  In this setting, crossing minimization
is interesting even if the order of the vertices along the spine is
given~\cite{BhoreGMN20:bookemb:JGAA,LiuCHW21:bookthickness:TCS}.

The 2-layer variant comes in two settings: \emph{one-sided crossing
    minimization} (\emph{OSCM}) and \emph{two-sided crossing minimization}
(\emph{TSCM}).  In OSCM, the input consists of a (bipartite) graph and
a linear order for the vertices on one side of the bipartition; the
task is to find a linear order for the vertices on the other side that
minimizes the total number of crossings.  In TSCM, the linear orders
on both layers can be chosen freely.
OSCM is an important step in the so-called Sugiyama framework for
drawing hierarchical graphs~\cite{sugiyama-etal-hierarchical-TSMC81},
that is, graphs where each vertex is assigned to a specific layer.
OSCM was the topic of the Parameterized Algorithms and Computational
Experiments Challenge
(PACE\footnote{\url{https://pacechallenge.org/2024/}}) 2024.
Both OSCM and TSCM are \NP-hard; OSCM even for the disjoint union
of 4-stars \cite{MunozUV01:GD:OSCM-hardess-4stars} and for
trees~\cite{dobler-arxiv23}.
On the positive side, OSCM admits a subexponential \FPT-algorithm; it runs in
$O(k2^{\sqrt{2k}}+n)$ time~\cite{kt-fssfpt-Algorithmica14}.
TSCM also admits an \FPT-algorithm; it runs in $2^{O(k)}+n^{O(1)}$
time~\cite{kobayashi-tamaki-IPL16}.

In the context of beyond-planar graphs, a local version of the
crossing number has also received considerable attention \cite{dujmovic-DagRep,survey-beyond-planarity}.  A graph is
\emph{$k$-planar} if it admits a drawing with at most $k$ crossings
per edge.  The \emph{local crossing number} of a graph is the
smallest~$k$ such that the graph is $k$-planar.
The recognition of 1-planar graphs has long been known to be
NP-hard~\cite{Grigoriev-Bodlaender-Algorithmica07}.
Later, it turned out that the recognition of $k$-planar graphs is
NP-hard for every~$k$~\cite{Urschel-Wellens-ILP21}.
Hence, it is unlikely that \FPT- or \XP-algorithms exist with respect
to the natural parameter~$k$.  On the other hand, recognizing 1-planar
graphs is fixed-parameter tractable with respect to tree-depth
and cyclomatic number~\cite{bannister-cabello-eppstein-JGAA18}.
The problem remains \NP-hard, however, for graphs of bounded bandwidth (and hence, pathwidth and treewidth).
The local crossing number has also been studied in the context of book
embeddings~\cite{LiuCH20:k-planar-bookemb:COCOA,ackssuw-ecog-SWAT24}.

In this paper, we study the above-mentioned geometric restrictions,
but with respect to the {\em local} crossing number.
The resulting graph classes are called \emph{2-layer $k$-planar}
graphs and \emph{outer $k$-planar} graphs; see \cref{fig:example-graph}.
\begin{figure}[h]
    \begin{subfigure}[b]{.18\textwidth}
        \centering
        \includegraphics[page=1]{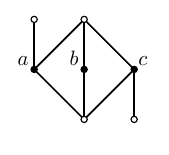}
        \subcaption{}
    \end{subfigure}
    \hfill
    \begin{subfigure}[b]{.2\textwidth}
        \centering
        \includegraphics[page=2]{example-graph}
        \subcaption{}
    \end{subfigure}
    \hfill
    \begin{subfigure}[b]{.2\textwidth}
        \centering
        \includegraphics[page=3]{example-graph}
        \subcaption{}
        \label{fig:example-graph-c}
    \end{subfigure}
    \hfill
    \begin{subfigure}[b]{.2\textwidth}
        \centering
        \includegraphics[page=4]{example-graph}
        \subcaption{}
    \end{subfigure}

    \caption{Drawings of the same bipartite graph with optimal local
        crossing number in different settings: (a)~planar drawing,
        (b)~2-layer 2-planar drawing without restriction,
        (c)~2-layer 3-planar drawing where the vertex order on the upper layer is
        fixed, (d)~outer 1-planar drawing.}
    \label{fig:example-graph}
\end{figure}
The former were studied by Angelini, Da Lozzo, F\"orster, and
Schneck~\cite{2layerkplanar}.  Among others, they gave bounds
on the edge density of these graphs and characterized 2-layer
$k$-planar graphs with the maximum edge density for $k \in \{2,4\}$.
They concluded that ``the general recognition and characterization
of 2-layer $k$-planar graphs remain important open problems''.
According to Schaefer's survey~\cite{schaefer-survey-EJC24} on
crossing numbers, Kainen~\cite{kainen-okp-CGTC89} introduced the
``local outerplanar crossing number'', which minimizes, over all
circular drawings, the largest number of crossings along any edge.
Outer $k$-planar graphs have been studied
by Pach and T\'oth~\cite{pach-toth-edge-density-Combinatorica97},
who showed that any outer $k$-planar graph with $n$
vertices has at most $4.1\sqrt{k}n$ edges.  For $k \le 3$, they
established a better bound $(k+3)(n-2)$, which is tight for
$k \in \{1,2\}$.  For $k \ge 5$, the constant factor was later
improved to $\sqrt{243/40} \approx 2.46$~\cite{OuterkPlanarBetterEdgeDensity}.

We study the
parameterized complexity of the two recognition problems with respect
to the natural parameter~$k$.  For $k=0$, the two classes of graphs
are exactly the caterpillars and outerplanar graphs, respectively,
which can be recognized in linear time.
There are also linear-time algorithms for recognizing outer $1$-planar
graphs~\cite{auer-etal-o1p-Algorithmica16,hong-etal-o1p-Algorithmica15}.
The authors of~\cite{hong-etal-o1p-Algorithmica15} posed the existence of
polynomial-time algorithms for recognizing outer 2-planar graphs
as an open problem.  A partial answer has been given by Hong
and Nagamochi~\cite{LinearTimeFullOuter2Planarity}, showing that full
outer 2-planar graphs can be recognized in linear time.
Outer $k$-planar drawings are \emph{full} if no crossing appears on
the boundary of the outer face.
The authors of~\cite{BeyondOuterplanarity} generalized this result
and showed that, for every integer~$k$, full outer $k$-planarity is testable in
$O(f(k)\cdot n)$ time, for a computable function~$f$.
They also showed that outer $k$-planar graphs can be
recognized in quasi-polynomial time,
which implies that, for every integer $k$, testing outer $k$-planarity is not \NP-hard unless the Exponential-Time Hypothesis fails.

\subparagraph*{Parameterized complexity}

We assume that the reader is familiar with basic concepts in parameterized complexity theory (see \cite{CyganFKLMPPS15:book,DowneyF13:book,FlumG06:book} for definitions of these concepts).
Zehavi~\cite{zehavi-survey-CSR22} gives a survey specifically
on the parameterized analysis of crossing minimization problems.
The class \emph{\XNLP} consists of all parameterized problems that can be solved non-deterministically in time $f(k)n^{O(1)}$ and space $f(k)\log n$, where $f$ is some computable function, $n$ is the input size, and $k$ is the parameter.
A parameterized problem $L_2 \subseteq \Sigma^* \times \mathbb N$ is said to be \emph{\XNLP-hard} if for any $L_1 \in \XNLP$, there is a \emph{parameterized logspace reduction} from $L_1$ to $L_2$, that is, there is an algorithm $\mathcal A$ and computable functions $f$ and $g$ that satisfy the following: Given $(x_1, k_1) \in \Sigma^*\times \mathbb N$, the algorithm $\mathcal A$ computes $(x_2, k_2) \in \Sigma^* \times \mathbb N$ such that $(x_1, k_1) \in L_1$ if and only if $(x_2, k_2) \in L_2$, $k_2 \le g(k_1)$, and $\mathcal A$ runs in space $O(f(k_1) + \log |x_1|)$.
A parameterized problem is said to be \emph{$\XNLP$-complete} if it is $\XNLP$-hard and belongs to $\XNLP$.
The class \XNLP~contains the
class~\W[$t$] for every $t \ge 1$~\cite{BodlaenderGNS21:FOCS:XNLP}.
Moreover, Pilipczuk and Wrochna~\cite{SliceWisePolynomialSpaceConjecture} conjectured that an \XNLP-hard problem does not admit an algorithm that runs in $n^{f(k)}$ time and $f(k) \cdot n^{O(1)}$ space for a computable function $f$, where $n$ is the size of the instance and $k$ is the parameter.
We refer to \cite{BodlaenderGNS21:FOCS:XNLP,ElberfeldST15-XNLP-Algorithmica}
for more information.

Recently, the authors of~\cite{bkkswz-colpp-SoCG24} showed the
first graph-drawing problem to be \XNLP-complete, namely
\emph{ordered level planarity}, parameterized by the number of levels.
Ordered level planarity is a restricted version of level planarity,
where for each level, the vertices on that level are given in order
(and the problem is to route the edges in a y-monotone and
crossing-free way).

\subparagraph*{Our contribution}

We present \XP-algorithms for recognizing
2-layer $k$-planar graphs and outer $k$-planar graphs, which implies
that both recognition problems can be solved in polynomial time
for every fixed~$k$; see \cref{sub:2SkP:XP,sub:OkP:XP}, respectively.
This solves the open problem regarding the recognition of outer 2-planar
graphs posed by the authors of \cite{hong-etal-o1p-Algorithmica15}.
We complement these results by showing that recognizing 2-layer
$k$-planar graphs is \XNLP-complete even for trees (\cref{sub:2SkP:XNLP-hardness}) and that recognizing outer
$k$-planar graphs is \XNLP-hard (\cref{sub:OkP:XNLP-hardness}).
This implies that both problems are \W$[t]$-hard for every~$t$~\cite{BodlaenderGNS21:FOCS:XNLP} and
that it is unlikely that they admit \FPT-algorithms.
On the other hand, we present an \FPT-algorithm for recognizing
2-layer $k$-planar graphs where the order of the vertices on one
layer is specified; see \cref{fig:example-graph-c} and \cref{sub:1SkP:FPT}.
We prove that two {\em edge-weighted} versions of this problem are \NP-hard; see \cref{sub:1SkP:hardness}.
Finally, we show that the local circular crossing number cannot be approximated even for graphs that are almost trees (that is, graphs with feedback vertex number~$1$); see \cref{sub:OkP:approx:hardness}.
We conclude with %
open problems; see \cref{sec:open}.

\section{Preliminaries}\label{sec:preliminaries}

Let $G$ be a graph.
We let $V(G)$ and $E(G)$ denote the sets of vertices and edges of~$G$, respectively.
For a vertex $v$ of~$G$, let $N_G(v)$ be the set of neighbors of~$v$ in~$G$,
and let $\delta_G(v)$ be the set of edges incident to $v$.
For $U \subseteq V(G)$, we define $N_G(U) = (\bigcup_{v \in U}N_G(v)) \setminus U$, $N_G[U] = N_G(U) \cup U$, and $\delta_G(U) = \bigcup_{v \in U}\delta_G(v)$.
We may omit the subscript~$G$ when it is clear from the context.
The subgraph of $G$ induced by $U \subseteq V(G)$ is denoted by $G[U]$.
A vertex is called a \emph{leaf} if it has exactly one neighbor.

We use $[\ell]$ as shorthand for $\{1,2,\dots,\ell\}$.
Let $n=|V(G)|$, and let $\sigma \colon V(G) \to [n]$ be a linear order of the vertices of~$G$ (i.e., a bijection between~$V(G)$ and $[n]$).
For $\{u, v\} \in E$, the \emph{stretch} of edge $\{u, v\}$ in $\sigma$ is defined as $|\sigma(u) - \sigma(v)|$.
The \emph{bandwidth} of $\sigma$ (with respect to $G$) is the maximum stretch of an edge of~$G$ in~$\sigma$.
The \emph{bandwidth} of~$G$, denoted by $\bw(G)$, is the minimum integer $k$ such that $G$ has a linear order $\sigma$ of $V$ with bandwidth at most~$k$.

We define a \emph{circular drawing} of a graph $G$ to be a cyclic order $D = (v_1, \dots, v_n)$ of $V(G)$.
We say that an edge $\{v_i, v_j\}$ with $i < j$ \emph{pierces} a pair of (not necessarily adjacent)
vertices $\{v_{i'}, v_{j'}\}$ with $i' < j'$ if either $1 \le i < i' < j < j' \le n$ or $1 \le i' < i < j' < j \le n$ holds.
In particular, if $\{v_{i'}, v_{j'}\}$ is an edge of $G$, we say that $\{v_i, v_j\}$ \emph{crosses} $\{v_{i'}, v_{j'}\}$.
For an edge~$e$, let \emph{$\cross_D(e)$} denote the number of edges that cross~$e$ in~$D$.
A circular drawing $D$ is \emph{$k$-planar} (or an \emph{outer $k$-planar drawing}) if every edge crosses at most $k$ edges in~$D$.
Since whether two edges cross is determined only by the cyclic vertex order, in this paper, we allow the edges to be drawn arbitrarily.

Let $G$ be a bipartite graph with $V(G)=X \cup Y$, $X \cap Y = \emptyset$, and $E(G) \subseteq X \times Y$.
A \emph{2-layer drawing} of $G$ is a pair $D =(<_X, <_Y)$ of (strict) linear orders~$<_X$ and~$<_Y$ defined on~$X$ and on~$Y$, respectively.
A \emph{crossing} in $D$ is defined by a pair of edges $\{x, y\}$ and $\{x', y'\}$ with distinct endpoints $x, x' \in X$ and distinct endpoints $y, y' \in Y$ such that $x <_X x'$ and $y' <_Y y$.
The notation $\cross_D(e)$ is defined as above.
Moreover, for two distinct edges $e$ and $e'$, let \emph{$\cross_D(e, e')$}$ = 1$ if $e$ crosses $e'$; $\cross_D(e, e') = 0$ otherwise.
For distinct $x, x' \in X$, we say that $x$ is \emph{to the left of} $x'$ in $D$ if $x <_X x'$.
Equivalently, we say that $x'$ is \emph{to the right of} $x$.
The \emph{leftmost} (resp. \emph{rightmost}) of $X$ in $D$ is the smallest (resp. largest) vertex in $X$ under the linear order $<_X$.
We also use these notions for vertices in~$Y$.
For an %
integer $k \ge 0$, a 2-layer drawing $D$ of $G$ is said to be \emph{$k$-planar} if, for each edge~$e$ of~$G$,
$\cross_D(e) \le k$.
For (not necessarily disjoint) vertex sets $A, B \subseteq X \cup Y$ and 2-layer drawings $D_A$ of $G[A]$ and $D_B$ of $G[B]$, we say that $D_A$ is \emph{compatible} with $D_B$ (or equivalently $D_B$ is compatible with $D_A$) if, for every pair $\{z,z'\} \subseteq A \cap B$ of vertices that both are in the same set of the bipartition $Z \in \{X,Y\}$, we have that $z <_Z z'$ in $D_A$ if and only if $z <_Z z'$ in $D_B$.

\section{Recognizing 2-Layer $k$-Planar Graphs -- The One-Sided Case}
\label{sec:2-layer:oskp}

In this section, we design an \FPT-algorithm for recognizing 2-layer $k$-planar graphs when the order of the vertices on one layer is given as input.
The problem is defined as follows.

\defproblem
{\OneSidedkPlanarity}
{A bipartite graph $(X \cup Y, E)$, an integer $k \geq 0$, and a linear order $<_X$ of~$X$.}
{Does $Y$ admit a linear order $<_Y$ such that $(<_X, <_Y)$ is a 2-layer $k$-planar %
    drawing of~$(X \cup Y, E)$?}

\subparagraph*{Degree reduction.}

Let $G = (X \cup Y, E)$ be a bipartite graph, and let $k$ be a non-negative integer.
We describe two simple reduction rules that yield an equivalent instance of \OneSidedkPlanarity where every vertex in~$X$ has degree at most $2k + 2$.

\begin{observation}\label{obs:2-layer:non-leaf-degree}
    Let $G = (X \cup Y, E)$ be a bipartite graph that contains a vertex with more than $2k+2$ non-leaf neighbors. Then $G$ is not 2-layer $k$-planar.
\end{observation}

\begin{proof}
    Suppose that $G$ admits a 2-layer $k$-planar drawing $D$.
    Let $v$ be a vertex of $G$ with more than $2k+2$ non-leaf neighbors.
    Without loss of generality, we assume that $v \in X$.
    Let $y_1, \dots, y_{d} \in Y$ be the non-leaf neighbors of $v$ appearing in this order in~$D$, and let $e_i$ denote the edge $\{v, y_i\}$ for $1 \le i \le d$.
    Since $y_{k + 2}$ is not a leaf, it has an incident edge other than~$e_i$.
    This edge has a crossing with either each of $e_1, \dots, e_{k+1}$ or each of $e_{k+3}, \dots, e_{d}$, contradicting the $k$-planarity of $D$.
    Hence, we have the following observation.
\end{proof}

\begin{lemma}\label{lem:2-layer:degree-reduction}
    Let $(G, <_X, k)$ be an instance of \OneSidedkPlanarity.
    If $G$ contains a vertex $v \in X$ with $\deg(v) > 2k + 2$ and with a leaf neighbor $y \in Y$, then
    $(G, <_X, k)$ is a \YES-instance if and only if $(G-y, <_X, k)$ is a \YES-instance.
\end{lemma}

\begin{proof}
    The forward implication is immediate since $G-y$ is a subgraph of $G$.

    Now suppose that $G-y$ has a 2-layer $k$-planar drawing $D' = (<_X, <_Y)$.
    Let $y_1, \dots, y_{d}$ be the neighbors of $v$ appearing in this order in $<_Y$, and, for $i \in [d]$, let $e_i = \{v, y_i\}$.
    Observe that $e_{k + 2}$ has no crossings, as otherwise there is an edge that has a crossing with either each of $e_1, \dots, e_{k+1}$ or each of $e_{k+3}, \dots, e_{d}$.
    Moreover, $y_{k + 2}$ does not have a neighbor in $X$ to the left of $v$, due to edges $e_1, \dots, e_{k+1}$.
    Thus, we can insert~$y$ immediately to the left of $y_{k + 2}$ in $D'$ without introducing a new crossing.
    Hence, the resulting drawing is a 2-layer $k$-planar drawing of $G$.
\end{proof}

Hence, in the following, we assume that every vertex in $X$ has degree at most $2k + 2$.

\subsection{An FPT-Algorithm}
\label{sub:1SkP:FPT}
In this subsection, we prove the following result.

\begin{theorem}\label{thm:2-layer:oskp:algo}
    \OneSidedkPlanarity can be solved in time $2^{O(k^2 \log k)}n^{O(1)}$, where $n$ is the number of vertices of the given graph.
    In other words, \OneSidedkPlanarity is fixed-parameter tractable when parameterized by~$k$.
\end{theorem}

Let $G= (X \cup Y, E)$ be the given graph, and let $<_X$ be the given order of the vertices in~$X$.
We assume that $G$ has no isolated vertices; otherwise, we simply remove them.
For a 2-layer drawing $D$ of a subgraph $G'$ of $G$, we say that $D$ \emph{respects} $<_X$ if, for every pair $\{x, x'\}$ of vertices of~$G'$, it holds that $x$ is to the left of $x'$ in $D$ if and only if $x <_X x'$.

Let $x_1, \dots, x_{|X|}$ be the vertices of $X$ in the order in which they appear in~$<_X$ and let $\ell = 4k + 1$.
If $|X| \leq \ell$, then $|Y| = O(\ell k) = O(k^2)$ and we can simply enumerate all the possible 2-layer drawings in time $2^{O(k^2\log k)}n^{O(1)}$.
Hence, we assume $|X| > \ell$.
For $i \in [|X| - \ell]$, let $X_{\le i} = \{x_1, \dots, x_{i+\ell}\}$ and $X_i = \{x_i, \dots, x_{i + \ell}\}$.
Correspondingly, let $G_{\le i} = G[N[X_{\le i}]]$ and $G_i = G[N[X_i]]$.
Our algorithm decides recursively whether $G_{\le i}$ admits a 2-layer $k$-planar drawing that extends a prescribed drawing~$D$ of~$G_i$.
To be more precise, let $D$ be a 2-layer $k$-planar drawing of $G_i$ that respects $<_X$.
Given $i \in [|X| - \ell]$, a drawing~$D$ of~$G_i$ respecting~$<_X$, and a function $\chi\colon \delta(X_i) \to \{0, \dots, k\}$, we define a Boolean value $\Dp(i, D, \chi)$ to be~\vtrue if and only if $G_{\le i}$ admits a 2-layer $k$-planar drawing \DDDD such that
\begin{itemize}
    \item \DDDD is compatible with $D$ and respects $<_X$, and
    \item for every edge $e \in \delta(X_i)$, it holds that $\chi(e) = \cross_{\DDDD}(e)$.
\end{itemize}
Our goal is to compute $\Dp(|X| - \ell, D, \chi)$ for every drawing
$D$ of $G_{|X| - \ell}$ and $\chi$.  Then $(G,<_X)$ is 2-layer $k$-planar if and only if $\Dp(|X| - \ell, D, \chi) = \vtrue$ for some proper choice of~$D$ and~$\chi$.
We compute the table $\Dp$ using the following dynamic programming algorithm.

For the base case $i = 1$, we compute the table entries by brute force:
for each possible 2-layer $k$-planar drawing $D$ of $G_1$ respecting $<_X$ and for every function $\chi\colon \delta(X_1) \to \{0, \dots, k\}$, we set $\Dp(1, D, \chi) = \vtrue$ if $\chi(e) = \cross_D(e)$ for every $e \in \delta(X_1)$; otherwise, we set $\Dp(1, D, \chi) = \vfalse$.
Now we show that, in any 2-layer $k$-planar drawing,
if two edges have endpoints in~$X$ that are far apart, then the edges do not cross.

\begin{lemma}\label{lem:2-layer:if-far-enough-they-dont-cross}
    Let $e = (x_p, y)$ and $f = (x_q, y')$ be two distinct edges.
    If $q \geq p + 2k$, then in every 2-layer $k$-planar drawing $D = (<_X, <_Y)$ of $G$, it holds that $y <_Y y'$ or that $y=y'$.
\end{lemma}

\begin{proof}
    Suppose otherwise, namely, $y >_Y y'$.
    Observe that, for every $t$ with $p < t < q$, there is at least one edge incident to vertex $x_t$, as $X$ contains no isolated vertices.
    Since each of these at least $2k-1$ edges crosses $e$ or $f$, at least one of $e$ or $f$ has more than $k$ crossings, which contradicts the $k$-planarity of $D$.
\end{proof}

This suggests the following recurrence.

\begin{lemma}\label{lem:2-layer:oskp:rec}
    Let $2 \le i \le |X| - \ell$, let $D_i$ be a 2-layer $k$-planar drawing of $G_i$ respecting $<_X$, and let $\chi_i\colon \delta(X_i) \to \{0, \dots, k\}$ such that $\chi_i(e) = \cross_{D_i}(e)$ for every $e \in \delta(x_{i + \ell})$.
    Then, $\Dp(i, D_i, \chi_i) = \vfalse$ if there is $y \in N(x_{i+\ell})$ such that $N(y)$ $\cap$ $\{x_{i+\ell-2k-1}, \dots, x_{i+\ell-1}\}$ $=$ $\emptyset$ but $N(y) \cap \{x_{1}, \dots, x_{i+\ell-2k-2}\} \neq \emptyset$.
    Otherwise,
    \begin{align*}
        \Dp(i, D_i, \chi_i) = \bigvee_{\DD, \xx} \Dp(i-1, \DD, \xx),
    \end{align*}
    where the \DD are taken over all 2-layer $k$-planar drawings of $G_{i-1}$ that are compatible with~$D_i$, and the $\xx$ are taken over all functions $\delta(X_{i-1}) \to \{0, \dots, k\}$ that satisfy
    \begin{align*}
        \xx(e) =\chi_i(e) - \sum_{f \in \delta(x_{i+\ell})}\cross_{D_i}(e, f)
    \end{align*}
    for every edge $e \in \delta(X_i) \cap \delta(X_{i-1})$.
\end{lemma}
\begin{proof}
    We first confirm that $\Dp(i, D_i, \chi_i) = \vfalse$ if there is a neighbor $y \in N(x_{i+\ell})$ such that $N(y) \cap \{x_{i+\ell-2k-1}, \dots, x_{i+\ell-1}\} = \emptyset$ but $N(y) \cap \{x_{1}, \dots, x_{i+\ell-2k-2}\} \neq \emptyset$.
    Let $x$ be an arbitrary vertex in $N(y) \cap \{x_{1}, \dots, x_{i+\ell-2k-2}\}$.
    As in the proof of~\cref{lem:2-layer:if-far-enough-they-dont-cross}, the edges incident to $\{x_{i+\ell-2k-1}, \dots, x_{i+\ell-1}\}$ must cross at least one of the edges $\{x, y\}$ and $\{x_{i+\ell}, y\}$, as they are not incident to $y$.
    Since there is no isolated vertex in $\{x_{i+\ell-2k-1}, \dots, x_{i+\ell-1}\}$, the number of such edges is at least $|\{x_{i+\ell-2k-1}, \dots, x_{i+\ell-1}\}| = (i+\ell-1) - (i+\ell-2k-1) + 1 = 2k + 1$.
    It follows that there is no 2-layer $k$-planar drawing of $G_{\leq i}$.
    In the following, we assume otherwise, namely, that there is no such $y$.

    Suppose that $\Dp(i, D_i, \chi_i) = \vtrue$.
    Then,
    there is a 2-layer $k$-planar drawing~\DDDD of~$G_{\le i}$ that is compatible with $D_i$ and, for every edge $e \in \delta(X_i)$, it holds that $\chi_i(e) = \cross_{\DDDD}(e)$.
    We need to show that
    there exists a triplet $t=(i-1,\DD,\xx)$ such that $\Dp(t)=\vtrue$.
    Let~\DDD and~\DD be subdrawings of~\DDDD induced by $G_{\le i-1}$ and $G_{i - 1}$, respectively.
    Then~\DDD is compatible with~\DD.
    By~\cref{lem:2-layer:if-far-enough-they-dont-cross}, edges incident to $x_{i + \ell}$ do not cross edges incident to $x_{i-1}$.
    Thus, for every edge $e \in \delta(x_{i-1})$, we have $\cross_{\DDDD}(e)=\cross_{\DDD}(e)$.
    Moreover, every edge $e \in \delta(X_i) \cap \delta(X_{i-1})$ has exactly $\chi_i(e) - \sum_{f \in \delta(x_{i+\ell})} \cross_{D_i}(e,f)$ crossings in~\DDD.
    Define~$\xx$ by setting $\xx(e) = \cross_{\DDD}(e)$ for every edge $e \in \delta(X_{i-1})$.
    Then, by definition, $\Dp(i-1, \DD, \xx) = \vtrue$
    since \DDD is a 2-layer $k$-planar drawing of $G_{\le i-1}$
    compatible with~\DD and,
    for every $e \in \delta(X_{i-1})$, it trivially holds that $\xx(e) = \cross_{\DDD}(e)$.

    Conversely, suppose that $\Dp(i - 1, \DD, \xx) = \vtrue$ for some $\DD$ and $\xx$.
    Let $D_{\le i - 1}$ be a 2-layer $k$-planar drawing of $G_{\le i - 1}$ that respects $<_X$ and is compatible with $\DD$ such that $\xx(e) = \cross_{D_{\le i - 1}}(e)$ for every edge $e \in \delta(X_{i-1})$.
    Now let us construct a 2-layer drawing $D_{\le i}$ from $D_{\le i-1}$ in the following manner:
    \begin{enumerate}
        \item append $x_{i+\ell}$ next to (i.e. to the right of) $x_{i+\ell-1}$;
        \item insert the missing vertices, i.e. $N(x_{i + \ell}) \setminus N(X_{\le i-1})$, according to the order in $D_i$ (we will show that the set of missing vertices is in fact $N(x_{i + \ell}) \setminus N(X_{i-1})$, and the way of inserting them is uniquely determined);
    \end{enumerate}
    We claim that this drawing $D_{\le i}$ is a 2-layer $k$-planar drawing of $G_{\le i}$ that respects $<_X$, is compatible with $D_i$, and satisfies $\chi_{D_{\le i}}(e) = \chi_i(e)$ for every edge $e \in \delta(X_i)$, which implies that $\Dp(i, D_i, \chi_i) = \vtrue$.
    It is clear that this is a 2-layer drawing of $G_{\le i}$ that respects $<_X$ and is compatible with $D_i$.
    Hence, it remains to show that $\chi_{D_{\le i}}(e) = \chi_i(e)$ for every edge $e \in \delta(X_i)$ and $D_{\le i}$ is $k$-planar.

    First, let us suppose that there exists a vertex $y \in N(x_{i+\ell}) \setminus N(X_{i-1})$ that appears in $D_{\le i-1}$, i.e. is contained in $N(X_{\le i-1})$.
    Since it is not adjacent to a vertex in $X_{i-1}$, it is adjacent to a vertex $x_j$ for some $j < i-1 \leq i+\ell-2k-1$.
    Hence, this contradicts our assumption from the first paragraph and there is no such vertex, which implies that the set of missing vertices is indeed $N(x_{i + \ell}) \setminus N(X_{i-1})$.

    We now show that the way of inserting them is uniquely determined.
    Let us take an arbitrary neighbor $y'$ of $x_{i+\ell-2k-2}$.
    Since $D_i$ is $k$-planar and $y'$ appears in $D_i$, by \cref{lem:2-layer:if-far-enough-they-dont-cross}, every missing vertex, which is adjacent to $x_{i+\ell}$, is to the right of $y'$ in $D_{i}$.
    Moreover, in $D_{\le i-1}$, every vertex to the right of $y'$ appears in $D_{i}$ by \cref{lem:2-layer:if-far-enough-they-dont-cross}, since otherwise it has a neighbor $x_j$ for some $j < i$ but $(i+\ell-2k-2) - (i-1) \geq 2k$.
    Hence, for every missing vertex, the position to insert it in $D_{\le i-1}$ is uniquely determined only by $D_i$.

    By the above discussion, we can also observe $\chi_{D_{\le i}}(e) = \chi_i(e)$ for every edge $e \in \delta(X_i)$.
    The unique way of inserting the missing vertices implies that only the edges in $D_i$ involve new crossings.
    Hence, $\chi_{D_{\le i}}(e) = \chi_{i-1}(e) + \sum_{f \in \delta(x_{i+\ell})}\cross_{D_i}(e, f) = \chi_i(e)$ follows for every edge $e \in \delta(X_i) \cap \delta(X_{i-1})$ and $\chi_{D_{\le i}}(e) = \cross_{D_i}(e) = \chi_i(e)$ for every $e \in \delta(x_{i + \ell})$.
    Observe that this discussion also shows the $k$-planarity of $D_{\le i}$.
\end{proof}

Our algorithm evaluates the recurrence of~\cref{lem:2-layer:oskp:rec} in a dynamic programming manner.
Note that if $\chi_i$ does not satisfy the condition in~\cref{lem:2-layer:oskp:rec}, then we can immediately conclude $\Dp(i, D_i, \chi_i) = \vfalse$ by~\cref{lem:2-layer:if-far-enough-they-dont-cross}.
To see the runtime bound, observe that, for each $i \in [|X| - \ell]$, the number of possible 2-layer $k$-planar drawings of~$G_i$ is upper bounded by $|N(X_i)|! \le ((\ell + 1) \cdot (2k+2))! = 2^{O(k^2\log k)}$ and the number of possible functions from $\delta(X_i)$ to $\{0,\dots,k\}$ is upper bounded by $(k+1)^{|\delta(X_i)|} = 2^{O(k^2\log k)}$.
Hence, we can evaluate the recurrence in time $2^{O(k^2\log k)}n^{O(1)}$, which completes the proof of \cref{thm:2-layer:oskp:algo}.
\subsection{NP-Hardness of the Weighted Version}
\label{sub:1SkP:hardness}

We can generalize \OneSidedkPlanarity to weighted settings.
Let $G = (X \cup Y, E)$ be a bipartite graph, and let $w \colon E \to \mathbb N_{>0}$ be an edge-weight function.
A 2-layer drawing $D$ of $(G, w)$ is said to be $k$-planar if, for each edge $e$ of~$G$, it holds that
\begin{align}\label{eq:weighted-crossing}
    \sum_{f \text{ crosses } e \text{ in } D} w(f) \le k.
\end{align}
It is straightforward to extend our algorithm to this weighted setting.
We show its \NP-hardness as follows\footnote{Recently, the \NP-hardness of \OneSidedkPlanarity itself has been shown~\cite{DBLP:journals/corr/abs-2510-00331}.}.

\begin{theorem}\label{thm:weighted-one-sided-np-hard}
    The weighted \OneSidedkPlanarity is (weakly) \NP-hard under~(\ref{eq:weighted-crossing}).
\end{theorem}

\begin{proof}
    We perform a polynomial-time reduction from \textsc{Partition}, which is (weakly)
    \NP-hard~\cite{garey1979computers}.  An instance of this problem is a
    set~$A=\{a_1, a_2, \dots, a_n\}$
    of $n$ positive integers, and the task is to partition~$A$ into two
    sets~$B$ and~$B'$ such that $\Sum(B) = \Sum(B')$, where $\Sum(S)$ denotes the sum of all integers in a set $S$.

    We construct a bipartite graph $G = (X \cup Y, E)$ and a linear order $<_X$ on $X$ as follows (see \cref{fig:weighted-one-sided-hardness-reduction}).
    The vertex set $X$ consists of $n + 2$ vertices $x_0, x_1, \dots, x_n, x_{n+1}$, which appears in this order on $<_X$.
    The other vertex set $Y$ consists of $n + 1$ vertices $y_{\mathrm{mid}}, y_1, \dots, y_n$.
    The edge set $E$ consists of $e_0 = \{x_0, y_{\mathrm{mid}}\}$, $e_{n+1} = \{x_{n+1}, y_{\mathrm{mid}}\}$ and $e_i = \{x_i, y_i\}$ for every $1 \leq i \leq n$.
    We set $w(e_0) = w(e_{n+1}) = 1$ and set $w(e_i) = 2 a_i$ for every $1 \le i \le n$.
    Lastly, we set $k$ to $\Sum(A) + 1$.

    \begin{figure}[h]
        \centering
        \includegraphics{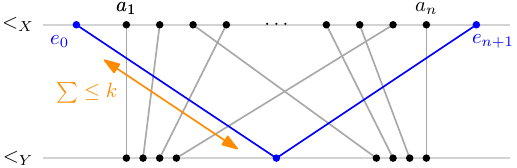}
        \caption{The graph $G$ and $<_X$ we construct. $\sum$ is the sum of weights of edges that cross $e_0$.}
        \label{fig:weighted-one-sided-hardness-reduction}
    \end{figure}

    Suppose that there is a partition $\{B, B'\}$ of $A$ such that $\Sum(B) = \Sum(B') = \Sum(A) / 2$.
    Let $j(1) < \dots < j(b)$ and $j'(1) < \dots < j'(b')$ be the indices of elements of $B$ and $B'$, respectively.
    We then construct $<_Y$ as
    \begin{align*}
        y_{j(1)} < \dots < y_{j(b)} < y_{\mathrm{mid}} < y_{j'(1)} < \dots <_Y y_{j'(b')}
    \end{align*}
    and show that $D = (<_X, <_Y)$ is a 2-layer $k$-planar drawing.
    We first consider the edge~$e_0$.
    It crosses exactly the edges $e_{j(1)}, \dots, e_{j(b)}$ and the sum of their weights is at most $\sum_{i=1}^{b} w(y_{j(i)}) = 2 \cdot \Sum(B) = \Sum(A) \leq k$.
    We can show the same bound for $e_{n+1}$.
    Next, we consider an edge $e_{j(i)}$ for some $1 \leq i \leq b$.
    It crosses $e_0$ and some of $e_{j'(1)}, \dots, e_{j'(b)}$, and it does not cross $e_{j(t)}$ for any $t$.
    Hence the sum of weights is at most $1 + \sum_{i=1}^{b'} w(y_{j'(i)}) = 1 + \Sum(A) = k$.
    We can show the same bound for $e_{j'(i)}$ and therefore $D$ is a 2-layer $k$-planar drawing.

    Conversely, suppose that there is a 2-layer $k$-planar drawing $D = (<_X, <_Y)$.
    Observe that, for any $1 \leq i \leq n$, the edge $e_i$ must cross either $e_0$ or $e_{n+1}$.
    Let $J$ denote the indices of edges that cross $e_0$ and $B = \{a_j \mid j \in J \}$.
    We define $J'$ and $B'$ in a similar manner with $e_{n+1}$.
    It is clear that $\{B, B'\}$ is a partition of $A$.
    By the $k$-planarity of $D$, considering $e_0$, $\sum_{j \in J} w(e_{j}) \leq k$ holds, which implies that $2 \cdot \Sum(B) \leq \Sum(A) + 1$.
    As $\Sum(A)$ must be an even number, $\Sum(B) \leq \Sum(A)/2$ holds.
    We also obtain a bound $\Sum(B') \leq \Sum(A)/2$ in the same way.
    Therefore, we have $\Sum(B) = \Sum(B')$.
\end{proof}

There is another reasonable definition of crossings in a weighted graph: A 2-layer drawing is defined to be $k$-planar if, for each edge $e \in E$, it holds that
\begin{align}\label{eq:another-weighted-crossing}
    \sum_{f \text{ crosses } e \text{ in } D} w(e) \cdot w(f) \le k.
\end{align}
By making $w(e_0)$ and $w(e_{n+1})$ sufficiently large,
a reduction similar to the one above will work.
\begin{remark}\label{thm:other-weighted-one-sided-np-hard}
    The weighted \OneSidedkPlanarity is (weakly) \NP-hard under~(\ref{eq:another-weighted-crossing}).
\end{remark}

\section{Recognizing 2-Layer $k$-Planar Graphs -- The Two-Sided Case}
\label{sec:2SkP}

The algorithm in \cref{thm:2-layer:oskp:algo} exploits the prescribed order $<_X$ on $X$, which is not specified in the two-sided case.
This difference is reflected by the parameterized complexity of the two problems.
The two-sided case turns out to be \XNLP-hard, meaning that it is unlikely to be fixed-parameter tractable.
On the other hand, we design a polynomial-time algorithm for the two-sided case, provided that $k$ is fixed.
We use this algorithm also to show that the problem is contained in \XNLP.
Formally, the problem is defined as follows.

\defproblem
{\TwoSidedkPlanarity}
{A bipartite graph $G = (X \cup Y, E)$ and an integer $k \geq 0$.}
{Does $G$ admit a 2-layer $k$-planar drawing?}

\subsection{An XP-Algorithm}
\label{sub:2SkP:XP}

To solve \TwoSidedkPlanarity, we extend the algorithm for \OneSidedkPlanarity presented in \cref{sec:2-layer:oskp}.
Let $G = (X \cup Y, E)$ be a bipartite graph, let $n$ be the number of vertices of~$G$, and let $k$ be a non-negative integer.
We assume that $G$ is connected; otherwise, the problem can be solved independently for each connected component.
Moreover, by applying \cref{obs:2-layer:non-leaf-degree} and applying
\cref{lem:2-layer:degree-reduction} first to~$X$ and then to~$Y$,
we assume that every vertex has degree at most $2k + 2$.
Our algorithm employs a dynamic programming approach analogous to that presented in \cref{sec:2-layer:oskp}.
Instead of a ``window'', we specify a subset $X_i \subseteq X$ of $\ell + 1$ vertices, which plays the same role as the window $\{x_i, \dots, x_{i+\ell}\}$.
However, this subset does not specify the graph $G_{\le i}$ on the left of the window, preventing us from defining the same type of subproblems as there.
We overcome this obstacle by applying an idea similar to that of Saxe~\cite{Saxe80:bandwidth:SIADM} for recognizing bandwidth-$k$ graphs.
To properly define the subproblems, we observe that $N[X_i]$ separates the subdrawings of the components of $G[V(G) \setminus N[X_i]]$ into left and right parts.

\begin{lemma}\label{lem:2-layer:tskp:separator}
    Let $D = (<_X, <_Y)$ be a 2-layer $k$-planar drawing of $G$, and let $S \subseteq X$ be a set of $\ell + 1$ vertices that appears consecutively in $D$.
    Let $x$ and $x'$ be the leftmost vertex and the rightmost vertex of~$S$ in~$D$, respectively.
    Then, for each component $C$ of $G[V(G) \setminus N[S]]$, the vertices in $C \cap X$ are either entirely to the left of $x$ or entirely to the right of $x'$.
\end{lemma}
\begin{proof}
    Suppose that $C$ has two vertices $u, v \in X \setminus S$ such that $u$ is to the left of $x$ and $v$ is to the right of $x'$ in $D$.
    Let $P$ be a path between $u$ and $v$ in $G[C]$.
    We can assume that $P$ has exactly two edges, $e$ and $f$.
    Observe that each edge incident to a vertex in $S$ crosses either $e$ or $f$.
    Since each vertex in $S$ has at least one incident edge, at least one of $e$ and $f$ must involve (at least) $(\ell + 1)/2 > k$ crossings.
\end{proof}

Suppose that $G$ has a 2-layer $k$-planar drawing $D = (<_X, <_Y)$.
For a family $\mathcal D \subseteq 2^{V(G)}$, we use $\mathcal D^X$ as shorthand for $\bigcup_{C \in \mathcal D} C \cap X$.
Let $x_1, \dots, x_{|X|}$ be the vertices of $X$ appearing in this order in $<_X$.
For $1 \le i \le |X| - \ell$, let $\mathcal C_i$ be the set of connected components in $G[V(G) \setminus N[\{x_i, \dots, x_{i+\ell}\}]]$.
By~\cref{lem:2-layer:tskp:separator}, we have $\mathcal C^X = \{x_1, \dots, x_{i-1}\}$ for some $\mathcal C \subseteq \mathcal C_i$.

Now, we can formally define our subproblems.
Let $S \subseteq X$ with $|S| = \ell + 1$, let $D$ be a 2-layer $k$-planar drawing of $G[N[S]]$, let $\chi\colon \delta(S) \to \{0, \dots, k\}$, and let $\mathcal C \subseteq \mathcal C_S$, where $\mathcal C_S$ is the set of components in $G[V(G) \setminus N[S]]$.
We define a Boolean value $\Dp(S, D, \chi, \mathcal C)$ to be $\vtrue$ if and only if there is a 2-layer $k$-planar drawing $D^*$ of $G[N[S \cup \mathcal C^X]]$ such that
\begin{itemize}
    \item $x <_X x'$ for each $x \in \mathcal C^X$ and each $x' \in S$,
    \item $D^*$ is compatible with $D$, and
    \item for every edge $e \in \delta(S)$, it holds that $\chi(e) = \cross_{D^*}(e)$.
\end{itemize}
Hence, $G$ has a 2-layer $k$-planar drawing if and only if $\Dp(S, D, \chi, \mathcal C_S)=\vtrue$ for some $S \subseteq X$, $D$, $\chi$, and $\mathcal C_S$.

To compute the values $\Dp(S, D, \chi, \mathcal C)$ for $S \subseteq X$ with $|S| = \ell + 1$, $D$, $\chi$, and $\mathcal C \subseteq \mathcal C_S$, we first compute the base cases where $\mathcal C = \emptyset$ and then the other cases in ascending order of $|S \cup \mathcal C^X|$.
This can be done by using a recurrence similar to the one in \cref{lem:2-layer:oskp:rec}.

To see the running time bound of the above algorithm, observe that the number of possible choices for $S$, $D$, $\chi$, and $\mathcal C$ is at most
\begin{align*}
    \sum_{S \subseteq X} |S|! \cdot |N(S)|! \cdot (k + 1)^{|\delta(S)|}\cdot 2^{|\mathcal C_S|} =
    n^{\ell + 1} \cdot 2^{O(k^2\log k)} \cdot 2^{|\mathcal C_S|}.
\end{align*}
The third factor can be bounded by $2^{O(k^3)}$ as follows:
Since $G$ is connected, each connected component of $G[V(G) \setminus N[S]]$ contains at least one vertex in $N(N(S)) \setminus S$.
This implies that the number of components in $G[V(G) \setminus N[S]]$ is at most $|N(N(S)) \setminus S| \le (2k+2)(2k+1)^2$.

\begin{theorem}\label{thm:2-layer:tskp:xp-alg}
    \TwoSidedkPlanarity can be solved in time $2^{O(k^3)}n^{4k + O(1)}$, that is, \TwoSidedkPlanarity is polynomial-time solvable when $k$ is fixed.
\end{theorem}

The above algorithm easily turns into a non-deterministic algorithm that runs in polynomial time and space $k^{O(1)} \log n$, which implies the following.

\begin{corollary}
    \label{cor:2-layer:tslp:xnlp}
    \TwoSidedkPlanarity is in \XNLP.
\end{corollary}

\begin{proof}
    We show that \TwoSidedkPlanarity can be solved in polynomial time and $k^{O(1)}\log n$ space.
    The idea of the algorithm is almost analogous to those used in \cite{BodlaenderGJJL22:IPEC:XNLP,BodlaenderGNS21:FOCS:XNLP}.
    This can be done by non-deterministically guessing table indices $S$, $D$, $\chi$, and $\mathcal C \subseteq \mathcal C_S$ of our dynamic programming and keeping track of table entries to check a certificate of a 2-layer $k$-planar drawing of $G$ without enumerating possible indices.
    It is easy to see that $S$, $D$, and $\chi$ are encoded with $k^{O(1)}\log n$ bits.
    Moreover, as seen in the proof of \cref{thm:2-layer:tskp:xp-alg}, $\mathcal C$ can be represented by a subset of $N(N(S))\setminus S$, which allows us to encode it with $k^{O(1)}\log n$ bits as well.
    Therefore, the algorithm runs in polynomial time and uses $k^{O(1)}\log n$ bits of space in total.
\end{proof}

\subsection{XNLP-Completeness}
\label{sub:2SkP:XNLP-hardness}

To complement the positive result in the previous subsection, we show that \TwoSidedkPlanarity is \XNLP-hard even on trees.
In contrast to our result, TSCM can be solved in polynomial time on trees~\cite{TwoSidedMinimizationOnTrees}.

\begin{theorem}\label{thm:two-sided-k-planar-w-hardness}
    \TwoSidedkPlanarity is \XNLP-complete w.r.t.~$k$ %
    even on trees.
\end{theorem}

\begin{proof}
    We perform a parameterized logspace reduction from \Bandwidth, where given a graph $G$ and an integer $b \geq 1$, the goal is to decide whether $\bw(G) \le b$.
    This problem is known to be \XNLP-hard when parameterized by~$b$, even on trees~\cite{Bandwidth-Wt-hard-Trees,BodlaenderGNS21:FOCS:XNLP}.

    Let $T$ be a tree with $n = |V(T)|$ and let $b$ be a non-negative integer.
    From the instance $(T, b)$ of \Bandwidth, we construct an instance $(G, k)$ of \TwoSidedkPlanarity as follows (see \cref{fig:two-layer-xp-hardness-reduction}).
    Let $\ell = 4b^2$.
    Starting with $T$, we subdivide each edge $e \in E(T)$ by introducing a vertex $w_{e}$.
    Then, for each original vertex $v$ in $T$, we add $\ell$ leaves that are adjacent to $v$.
    The leaves that are added in the above construction are called \emph{pendant vertices}, and the edges incident to the pendant vertices are called \emph{pendant edges}; other edges are called \emph{non-pendant edges}.
    We let $G$ denote the graph obtained from $T$ in this way and set $k \coloneqq \ell(b-1)/2 + 4b - 4$.
    Observe that the graph $G$ is also a tree.
    Let $X = V(T)$ and $Y = V(G) \setminus X$.
    It is not hard to verify that the construction of $G$ can be done in polynomial time and $O(\log n)$ space, as we only subdivide each edge of $T$ once and add pendant vertices.

    \begin{figure}[h]
        \centering
        \includegraphics{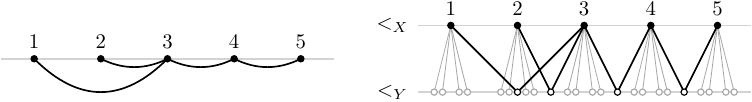}
        \caption{A minimum-bandwidth order of $T$ and the 2-layer drawing of $G$ that we construct.}
        \label{fig:two-layer-xp-hardness-reduction}
    \end{figure}

    Suppose that $G$ has a 2-layer $k$-planar drawing $D = (<_X, <_Y)$.
    We claim that the bandwidth of the linear order $\sigma$ on $V(T)$ naturally obtained from $<_X$ is at most~$b$.
    Suppose for a contradiction that there exists an edge $e \in E(T)$ whose stretch with respect to $\sigma$ exceeds $b$.
    Then, in the drawing $D$, the two edges incident to $w_e$ cross at least $\ell b$ pendant edges in total.
    This implies that $2k \ge \ell b$.
    However,
    \begin{align*}
        2k = \ell(b-1) + 8b - 8 = \ell b - 4b^2 + 8b - 8 = \ell b - 4(b-1)^2 - 4 < \ell b,
    \end{align*}
    which leads to a contradiction.

    Conversely, suppose that $\bw(T) \le b$.
    Let $\sigma$ be a linear order of $V(T)$ with bandwidth at most~$b$.
    We construct a 2-layer drawing $D = (<_X, <_Y)$ of $G$ from $\sigma$ and then show that $D$ is $k$-planar.
    We set $<_X$ to the linear order obtained from $\sigma$ and sort the pendant vertices in $Y$ according to the order of their neighbors in~$\sigma$.
    For each $\{u, v\} \in E(T)$, let $P_{uv}$ be the set of pendant edges incident to some vertex $x$ satisfying $u <_X x <_X v$.
    We then insert the vertex $w_{\{u, v\}}$ so that $\{u, w_{\{u, v\}}\}$ and $\{v, w_{\{u, v\}}\}$ cross exactly the same number of pendant edges of $P_{uv}$.
    This can be done as $P_{uv}$ has an even number of pendant edges.
    Note that no pendant edges outside of $P_{uv}$ cross either $\{u, w_{\{u, v\}}\}$ or $\{v, w_{\{u, v\}}\}$.
    The above construction yields a 2-layer drawing $D$ of $G$, and we show that $D$ is $k$-planar.

    We first consider a pendant edge $e$, which is incident to a vertex $x \in X$.
    Clearly, $e$ does not cross any other pendant edges.
    Moreover, for $f = \{u,v\} \in E(T)$, $e$ crosses exactly one of $\{u, w_{f}\}$ and $\{v, w_{f}\}$ incident to $w_{f}$ if $u <_X x <_X v$; $e$ never crosses other non-pendent edges.
    As the stretch of $f$ is at most $b$, we have $\sigma(x) - \sigma(u) \le b - 1$ and $\sigma(v) - \sigma(x) \le b - 1$.
    Now, consider the vertex set $S_x = \{x' : |\sigma(x) - \sigma(x')| \le b - 1\}$.
    Since $T$ is a tree, $T[S_x]$ is a forest, and hence $T[S_x]$ has at most $|S_x| - 1 = 2b-2$ edges.
    Thus, $e$ crosses at most $2b-2 \leq 4b-4$ non-pendant edges in $D$.

    We next consider a non-pendant edge incident to a vertex $w_{e} \in Y$ for some $e = \{u, v\}$.
    We only count the number of crossings involving $f \coloneqq \{u, w_e\}$ as the other case is symmetric.
    The edge $f$ crosses exactly $|P_{uv}|/2 \le \ell(b-1)/2$ pendant edges.
    Similarly to the previous case, every non-pendant edge that crosses $f$ is a part of a subdivided edge that connects two vertices in $S_u = \{u' : |\sigma(u) - \sigma(u')| \le b - 1\}$.
    The number of such non-pendant edges is $2 (|S_u| - 1) = 4b - 4$.
    Hence, there are at most $\ell(b-1)/2 + 4b - 4 \le k$ crossings involving $f$ in $D$.
    Therefore, $D$ is $k$-planar.
\end{proof}

\section{Recognizing Outer $k$-Planar Graphs}
\label{sec:OkP}

In this section, we discuss the parameterized complexity of recognizing outer $k$-planar graphs.

\defproblem
{\OuterkPlanarity}
{A graph $G$ with $n$ vertices and an integer $k$.}
{Does $G$ admit an outer $k$-planar drawing?}

\subsection{An XP-Algorithm}
\label{sub:OkP:XP}

In this subsection, we show our main result, an \XP-algorithm for \OuterkPlanarity with respect to~$k$.
Note that a graph is outer $k$-planar if and only if its biconnected components are outer $k$-planar;
this can be shown in a similar manner as \cite[Theorem 4]{LinearTimeFullOuter2Planarity} for $k=2$.
Hence, we assume the input graph to be biconnected.

\begin{theorem}\label{thm:outer-k-planarity-xp}
    \OuterkPlanarity can be solved in time $2^{O(k \log k)}n^{3k + O(1)}$, that is, \OuterkPlanarity is polynomial-time solvable when $k$ is fixed.
\end{theorem}

Let $G$ be the input graph, let $\vec{e} = (u, v)$ be an ordered pair of two distinct vertices of $G$, and let~$R$ be a subset of $V(G) \setminus \{u, v\}$ such that there are at most $k$ edges between $R$ and $L$, where $L = V(G) \setminus (\{u, v\} \cup R)$.
Let $C$ be the set of these edges.
Let $\tau$ be a linear order of $C$, let $c_1, \dots, c_\ell$ denote the vertices of $C$ in the order given by $\tau$, and let $\chi \colon C \to \{0, \dots, k\}$.
Let $G_{\tau, \vec{e}, R}$ be the graph obtained by adding $\ell$ vertices $t^\tau_1, t^\tau_2, \dots, t^\tau_\ell$ to the induced subgraph $G[\{u, v\} \cup R]$ and by connecting $t^\tau_i$ and the endpoint of $c_i$ in $R$ for every $i \in [\ell]$.
Then we define a Boolean value $\Dp(\vec{e}, R, \tau, \chi)$ to be \vtrue if and only if $G_{\tau, \vec{e}, R}$ admits an outer $k$-planar drawing $D$ with the following properties:
\begin{enumerate}[(P1)]
    \item \label{enum:cyclic} the cyclic order of~$D$ contains
          $(u, t^\tau_1, t^\tau_2, \dots, t^\tau_\ell, v)$ as a consecutive
          subsequence, and
    \item \label{enum:chi} for every edge $c_i \in C$, it holds that
          $\chi(c_i) = \cross_D(c_i)$.
\end{enumerate}
Clearly, the graph $G$ admits an outer $k$-planar drawing if and only if there exists a vertex pair~$\vec{e}$ such that $\Dp(\vec{e}, V(G) \setminus \{u, v\}, f_\emptyset, f_\emptyset)=\vtrue$, where $f_\emptyset$ is the empty function.

We evaluate the recurrence as follows.
For every base case, namely where $R = \emptyset$, $\Dp(\vec{e}, R, \tau, \chi)$ is $\vtrue$ since $C$ is also empty.

When $R \neq \emptyset$, we compute $\Dp(\vec{e}, R, \tau, \chi)$ for smaller sets of type~$R$.
In this case, with the same technique as that used in \cite[Lemma 6]{OuterkPlanarTriangulation}, we can show the following.

\begin{lemma}\label{lem:outer-k-planarity-existence-of-w}
    If $G_{\tau, \vec{e}, R}$ admits an outer $k$-planar drawing $D$
    with properties P\ref{enum:cyclic} and P\ref{enum:chi}, there is
    a vertex $w \in R$ such that vertex pairs $\{u, w\}$ and $\{v, w\}$ are
    pierced by at most $k$ edges in $D$.
\end{lemma}

\begin{proof}
    The claim can be shown by following the proof of \cite[Lemma 6]{OuterkPlanarTriangulation}.
    In that lemma, the authors considered a maximal outer $k$-planar graph~$G$ with $n$ vertices and its outer $k$-planar drawing $D_G = (v_1, \dots, v_n)$.
    By maximality, $G$ contains the edge $\{v_i, v_{i+1}\}$ for every $i \in [n-1]$ and the edge $\{v_n, v_1\}$.
    The authors called the cycle consisting of these edges the \emph{outer cycle}.
    They showed that the outer cycle admits a triangulation such that each edge of the triangulation is pierced by at most $k$ edges in~$D_G$.

    The authors showed the existence of such a triangulation by showing that if the vertex pair $\{v_i, v_r\}$ with $i + 1 < r$, which they call an \emph{active link} in the proof, is pierced by at most $k$ edges in $D_G$, then there exists an index~$j$ with $i < j < r$ such that both $\{v_i, v_j\}$ and $\{v_j, v_r\}$ are pierced by at most $k$ edges in~$D$.
    As $\{v_1, v_n\}$ is not pierced, starting from $\{v_1, v_n\}$, we can recursively construct a desired triangulation.

    Since they did not use the maximality of~$G$ to show the existence of such an index~$j$, we can apply the proof directly.
    Property~P\ref{enum:cyclic}, which requires the cyclic order of $D$ to contain $(u, t^\tau_1, t^\tau_2, \dots, t^\tau_\ell, v)$ as a consecutive subsequence, assures that $\{u, v\}$ has $\ell \leq k$ piercing edges.
    Hence, by treating~$v$ and~$u$ as~$v_i$ and~$v_r$, respectively, we obtain in the same manner a vertex $w \in R$ ($= v_j$) such that $\{u, w\}$ and $\{v, w\}$ are also pierced by at most $k$ edges in~$D$.
\end{proof}

Hence, we can compute $\Dp(\vec{e}, R, \tau, \chi)$ by checking all the ways to split the instance at the vertex~$w$.
Let $\{R_1, R_2\}$ be a partition of $R \setminus \{w\}$, let $L_1 = V(G_{\tau, \vec{e}, R}) \setminus (\{u, w\} \cup R_1)$ and let $L_2 = V(G_{\tau, \vec{e}, R}) \setminus (\{v, w\} \cup R_2)$.
For $i \in [2]$, let $C_i$ be the set of edges between $R_i$ and $L_i$, let $\ell_i = |C_i|$, let $\tau_i$ be a linear order of $C_i$, and let $\chi_i \colon C_i \to \{0, \dots, k\}$ be a function.
We say that $w, R_1, \tau_1, \chi_1, R_2, \tau_2, \chi_2$ are \emph{consistent}
if the following holds ($\cross_{\tau, \tau_1, \tau_2}$ is defined below):
\begin{itemize}
    \item there are at most $k$ edges between $L_1$ and $R_1$ and at most $k$ edges between~$L_2$ and~$R_2$,
    \item for every edge $c$ between $L$ and $R_i$, $\chi_i(c) = \chi(c) - \cross_{\tau, \tau_1, \tau_2}(c)$ holds for $i \in [2]$,
    \item for every edge $c$ between $L$ and $\{w\}$, $\chi(c) = \cross_{\tau, \tau_1, \tau_2}(c)$,
    \item for every edge $c$ between $\{v\}$ and $R_1$, $\chi_1(c) + \cross_{\tau, \tau_1, \tau_2}(c) \leq k$,
    \item for every edge $c$ between $\{u\}$ and $R_2$, $\chi_2(c) + \cross_{\tau, \tau_1, \tau_2}(c) \leq k$, and
    \item for every edge $c$ between $R_1$ and $R_2$, $\chi_1(c) + \chi_2(c) + \cross_{\tau, \tau_1, \tau_2}(c) \leq k$.
\end{itemize}
Informally, the value $\cross_{\tau, \tau_1, \tau_2}(c)$ is the number of crossings on $c$ inside the ``triangle'' consisting of $\{u, v, w\}$.
To define it formally, let us consider a circular drawing $D_H = (u, t^\tau_{1}, \dots, t^\tau_{\ell}, v, t^{\tau_2}_{\ell_2}, \dots t^{\tau_2}_{1}, w, t^{\tau_1}_{\ell_1}, \dots, t^{\tau_1}_{1})$ of a graph $H$.
For each edge $c \in (E \cap \{u, v\}) \cup C \cup C_1 \cup C_2$, the graph $H$ contains an edge $f(c)$ defined as follows.
If $c = \{u, v\}$, $f(c)$ simply connects $u$ and $v$.
Suppose that $c$ is incident to exactly one vertex $x \in \{u, v, w\}$.
This implies that $c$ is contained in exactly one of $C$, $C_1$, and $C_2$, which also means that $c$ is contained in the domain of exactly one $\tau' \in \{\tau, \tau_1, \tau_2\}$.
Then $f(c)$ connects $x$ and $t^{\tau'}_{\tau'(c)}$.
Otherwise, $c$ is contained in exactly two of $C$, $C_1$, and $C_2$, since $c$ is not contained in $C \cap C_1 \cap C_2$.
Similarly to the previous case, $c$ is contained in the domains of distinct $\tau', \tau'' \in \{\tau, \tau_1, \tau_2\}$.
Then $f(c)$ connects $t^{\tau'}_{\tau'(c)}$ and $t^{\tau''}_{\tau''(c)}$.
Now we define $\cross_{\tau, \tau_1, \tau_2}(c) = \cross_{D_H}(f(c))$.

We are ready to state \cref{lem:outer-k-planarity-dp}, which formalizes
the above idea of splitting an instance $((u,v), R, \tau, \chi)$
at a vertex~$w$ in~$R$ into two subinstances $((u,w), R_1, \tau_1, \chi_1)$
and $((w,v), R_2, \tau_2, \chi_2)$; see
\cref{fig:outer-k-planarity-xp-image}, where some edges are curved for better visualization.

\begin{figure}[h]
    \centering
    \includegraphics{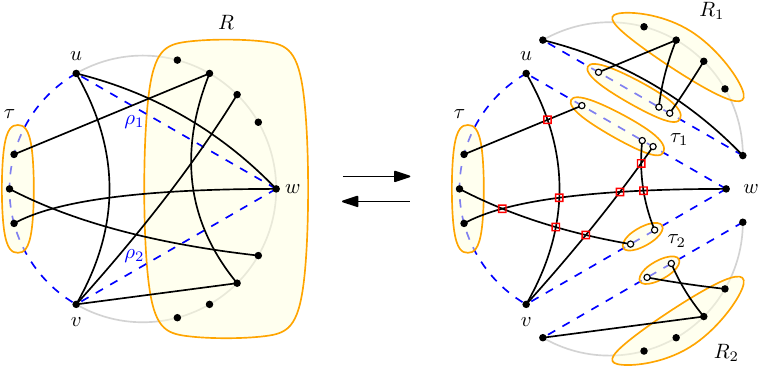}
    \caption{An image of \cref{lem:outer-k-planarity-dp}. The red boxes are the crossings considered in $\cross_{\tau, \tau_1, \tau_2}$.}
    \label{fig:outer-k-planarity-xp-image}
\end{figure}

\begin{lemma}\label{lem:outer-k-planarity-dp}
    For $\vec{e} = (u,v)$, it holds that
    \begin{align*}
        \Dp(\vec{e}, R, \tau, \chi) = \bigvee_{\substack{w, R_1, \tau_1, \chi_1, R_2, \tau_2, \chi_2 \\ \text{consistent}}} \Dp((u, w), R_1, \tau_1, \chi_1) \;\land\; \Dp((w, v), R_2, \tau_2, \chi_2).
    \end{align*}
\end{lemma}

\begin{proof}
    Suppose that $\Dp(\vec{e}, R, \tau, \chi) = \vtrue$, that is, there
    is an outer $k$-planar drawing
    $D = (u, t^{\tau}_{1}, \dots, t^{\tau}_{\ell}, v, v_1, \dots, v_r)$
    of $G_{\tau, e, R}$ that satisfies
    properties~P\ref{enum:cyclic}~and~P\ref{enum:chi}.
    We assume that edges are drawn as straight-line segments in $D$.
    By \cref{lem:outer-k-planarity-existence-of-w} there is a vertex $w \in R$ for some $w = v_i$ such that both $\{u, w\}$ and $\{v, w\}$ have at most $k$ piercing edges.
    The vertices $u, v, w$ divide the circumference into the three arcs $\rho_{\bar{u}}, \rho_{\bar{v}}, \rho_{\bar{w}}$, where $\rho_{\bar{x}}$ is the arc between vertices other than $x$ that does not pass through $x$ for each $x \in \{u, v, w\}$.
    Then, following the line through $u$ and $w$, we can take a curve $\rho_1$ between $u$ and $w$, inside the circle, such that it crosses exactly the piercing edges, does not pass through any crossing between piercing edges, and separates $v$ and the segment representing the edge $\{u, w\}$ (if it exists).
    We can take a curve $\rho_2$ between $v$ and $w$ similarly.
    Let $R_1 = \{v_{i+1}, \dots, v_r\}$ and $R_2 = \{v_1, \dots, v_{i-1}\}$.
    We then define a linear order $\tau_1$ on the edges between $R_1$ and $L_1 \coloneqq V(G_{\tau, \vec{e}, R}) \setminus (\{u, w\} \cup R_1)$ in such a way that $\tau_1$ orders those edges in ascending order of the distance between $u$ and the crossing with the curve $\rho_1$.
    Similarly, $\tau_2$ is defined in such a way that $\tau_2$ orders the edges $R_2$ and $L_2 \coloneqq V(G_{\tau, \vec{e}, R}) \setminus (\{v, w\} \cup R_2)$ in ascending order of the distance between $w$ to the crossing with the curve $\rho_2$.
    If we cut the drawing~$D$ along the curves as \cref{fig:outer-k-planarity-xp-image}, the drawing $D$ can be decomposed into three subdrawings $D_H$, $D_1$, and $D_2$: $D_H$ is the drawing inside the region surrounded by arcs $\rho_{\bar{w}}$, $\rho_2$, and $\rho_1$; $D_1$ is the drawing inside the region surrounded by arcs $\rho_{\bar{v}}$ and $\rho_1$; $D_2$ is the drawing inside the region surrounded by arcs $\rho_{\bar{u}}$ and $\rho_2$.
    Since each crossing on the edges between $R_1$ and $L_1$ is contained in exactly one of $D_H$, $D_1$, $D_2$, we have $\chi_1(c) = \chi(c) - \cross_{\tau, \tau_1, \tau_2}(c)$ for each edge $c$ between $L$ and $R_1$, and we have $\chi_1(c) + \cross_{\tau, \tau_1, \tau_2}(c) \le k$ for each edge $c$ between $\{v\}$ and $R_1$.
    As $D_1$ is a circular drawing of $G_{\tau_1, (u, w), R_1}$ that
    satisfies~P\ref{enum:cyclic} and~P\ref{enum:chi}, we have
    $\Dp((u, w), R_1, \tau_1, \chi_1) = \vtrue$, and similarly, we have
    $\Dp((w, v), R_2, \tau_2, \chi_2) = \vtrue$.
    Since each edge $c$ between $R_1$ and $R_2$ satisfies $\chi_1(c) + \chi_2(c) + \cross_{\tau, \tau_1, \tau_2}(c) \le k$, we conclude that $w, R_1, \tau_1, \chi_1, R_2, \tau_2, \chi_2$ are consistent.

    Suppose that there are consistent $w$, $R_1$, $\tau_1$, $\chi_1$,
    $R_2$, $\tau_2$, $\chi_2$ such that
    $\Dp((u, w), R_1, \tau_1, \chi_1) = \Dp((w, v), R_2, \tau_2, \chi_2)
        = \vtrue$.  Let~$D_1$ and~$D_2$ be circular drawings of
    $G_{\tau_1, (u, w), R_1}$ and $G_{\tau_2, (w, v), R_2}$,
    respectively, that satisfy~P\ref{enum:cyclic} and~P\ref{enum:chi}.
    Let $\sigma_1 = (w, \dots, u)$ and $\sigma_2 = (v, \dots, w)$ be the linear orders of $\{w, u\} \cup R_1$ and $\{v, w\} \cup R_2$ obtained from $D_1$ and $D_2$ by removing the vertices $t^{\tau_1}_i$ and $t^{\tau_2}_j$ for $i \in [\ell_1]$ and $j \in [\ell_2]$, respectively.
    Then we obtain a cyclic order $\sigma$ of $G_{\tau, e, R}$ by concatenating $\sigma_2, \sigma_1, (u, t^\tau_1, \dots, t^\tau_\ell, v)$ in this order, identifying the two occurrences of each of $w, u, v$ with each other.
    It is not difficult to see that, by combining $D_H$, $D_1$, and
    $D_2$ as in \cref{fig:outer-k-planarity-xp-image}, we obtain a
    drawing~$D$ with linear order~$\sigma$ that
    satisfies~P\ref{enum:cyclic} and~P\ref{enum:chi}.
    In other words, $\Dp(\vec{e}, R, \tau, \chi)=\vtrue$.
\end{proof}

Na\"ively, the number of $R$'s to consider is $\Theta(2^{n})$, which does not give an \XP-algorithm.
However, the following lemma assures that we do not need to consider so many sets of type~$R$. %

\begin{lemma}\label{lem:outer-k-planarity-xp-separator-size}
    Let $G$ be a biconnected graph that admits an outer $k$-planar drawing~$D$.
    Let $\{u, v\}$ be a pair of distinct vertices of $G$ that has at most $k$ piercing edges in $D$.
    Then the number of $R$'s such that $\Dp((u, v), R, \tau, \chi) = \vtrue$ for some $\tau$ and $\chi$ is at most $2^{O(k)}m^{k+O(1)}$, where $m = |E(G)|$.
    Moreover, such $R$'s can be enumerated in $2^{O(k)}m^{k+O(1)}$ time.
\end{lemma}

\begin{proof}
    We first bound the number of disjoint paths between two vertices in $G$.

    \begin{claim}\label{obs:outer-k-planar-number-of-disjoint-paths}
        Let $G$ be an outer $k$-planar graph, and let $u$ and~$v$ be two
        distinct vertices of~$G$.
        Then, there are at most $2k + 3$ (internally) vertex-disjoint paths between $u$ and $v$ in $G$.
    \end{claim}

    \begin{claimproof}
        Let $D = (v_1, \dots, v_n)$ be an outer $k$-planar drawing of $G$.
        In the following, we assume that $u = v_1$ and $v = v_i$ for some $i$.

        We first consider the case where there is an edge $\{l, r\}$ that pierces $\{u, v\}$ in $D$.
        Then, observe that each path between $u$ and $v$ that contains neither $l$ nor $r$ must cross the piercing edge $\{l, r\}$.
        Due to the $k$-planarity of $D$, there can be at most $k$ such paths, and hence, there are at most $k + 2$ vertex-disjoint paths between $u$ and $v$ in $G$.

        Suppose otherwise that no edge pierces $\{u, v\}$ in $D$.
        We say that a path is \emph{non-trivial} if it has at least two edges.
        Observe that each non-trivial path between $u$ and $v$ is contained in either $X \coloneqq \{v_1, v_2, \dots, v_i\}$ or $Y \coloneqq \{v_1, v_n, \dots, v_i\}$ since there is no piercing edge.
        Suppose that there are $k + 2$ disjoint non-trivial paths between $u$ and $v$ in $G[X]$.
        Let $v_j$ be the neighbor of $u = v_1$ in one of these paths such that all the other $k + 1$ neighbors are between $v_2$ and $v_{j-1}$.
        Since these paths are non-trivial, the $k + 1$ paths other than the one staring with $\{u, v_j\}$ must cross the edge $\{u, v_j\}$, contradicting the $k$-planarity of $D$.
        Thus, there are at most $k + 1$ disjoint non-trivial paths between $u$ and $v$ in $G[X]$.
        By applying the same argument to $G[Y]$, there are at most $2k + 2$ disjoint non-trivial paths between $u$ and $v$ in $G$, which implies the claimed upper bound.
    \end{claimproof}

    We now turn to the bound on the number of $R$'s such that $\Dp((u, v), R, \tau, \chi) = \vtrue$ for some valid $\tau$ and $\chi$.
    To this end, we first remove the vertices~$u$ and~$v$ from $G$ and let $H$ be the remaining graph.
    Let $H_1, H_2, \dots, H_c$ be the connected components of~$H$.
    Since $G$ is biconnected, we have $N_G(H_i) = \{u, v\}$ for every $i \in [c]$.
    Moreover, by \cref{obs:outer-k-planar-number-of-disjoint-paths}, there are at most $2k + 3$ vertex-disjoint paths between $u$ and $v$.
    Hence, we have $c \le 2k + 3$.

    Let $H'_1, \dots, H'_d$ be the connected components of the graph obtained from $H$ by deleting the edges $e_1, \dots, e_\ell$ that pierce $\{u,v\}$.
    Since each $R$ that is separated from $L = V(G) \setminus (\{u, v\} \cup R)$ in $H$ by removing the piercing edges $\{e_1, \dots, e_\ell\}$ is a union of these components, there are $2^{d}$ possibilities for such $R$'s.
    Clearly, for every $i \in [\ell]$, the piercing edge $e_i$ connects at most two of the components.
    Hence there are at most $2k$ components that contain at least one end vertex of a piercing edge.
    Moreover, for each $H'_i$ that does not contain an end vertex of a piercing edge, we have $H'_i = H_j$ for some $j$.
    In other words, by not only removing $u$ and $v$ but also the $\ell \le k$ edges piercing $\{u,v\}$, the number of resulting components increases by at most $2k$.
    Therefore, we have $d \le c + 2k$, which implies that the number of possible $R$'s is at most
    \begin{align*}
        \sum_{\ell = 0}^{k} \binom{m}{\ell} \cdot 2^{d} = 2^{O(k)} m^{k + O(1)}.
    \end{align*}

    The above argument readily turns into an algorithm for enumerating such $R$'s in time $2^{O(k)} m^{k + O(1)}$ as well.
\end{proof}

With \cref{lem:outer-k-planarity-xp-separator-size} and the fact that $m = O(\sqrt{k} n)$~\cite{pach-toth-edge-density-Combinatorica97}, the number of combinations of arguments $\{e, R, \sigma, \chi\}$ to consider is at most
\begin{align*}
    n^2 \cdot 2^{O(k)} m^{k + O(1)} \cdot k! \cdot (k+1)^k
     & =  2^{O(k \log k)} n^{k + O(1)}.
\end{align*}
To compute the value $\Dp(\vec{e}, R, \sigma, \chi)$ as in \cref{lem:outer-k-planarity-dp}, we guess at most
\begin{align*}
    n \cdot (2^{O(k \log k)} n^{k + O(1)})^2 = 2^{O(k \log k)} n^{2k + O(1)}
\end{align*}
possible combinations of $w, R_1, \tau_1, \chi_1, R_2, \tau_2, \chi_2$.
For each guess, checking the consistency takes $n^{O(1)}$ time.
Hence, the total running time to fill the table is $2^{O(k \log k)} n^{3k + O(1)}$.
This completes the proof of \cref{thm:outer-k-planarity-xp}.

\subsection{NP-Hardness of Approximation}
\label{sub:OkP:approx:hardness}

In this subsection, we show an inapproximability result for \OuterkPlanarity even for graphs that are almost trees, whereas trees can be drawn without any crossings.

\begin{theorem}\label{thm:outer-k-planarity-np-hard}
    For any fixed $c \ge 1$, there is no polynomial-time $c$-approximation algorithm for \OuterkPlanarity unless $\cP = \NP$, even for graphs with feedback vertex number~$1$.
\end{theorem}

Our proof is by reduction from \Bandwidth on trees, which is \NP-hard to approximate within any constant factor~\cite{HardnessApproximatingBandwidth}.
In other words, given a tree~$T$, there is no polynomial-time algorithm to distinguish between the cases $\bw(T) \le b$ and $\bw(T) > cb$ for any constant $c \ge 1$, unless $\cP = \NP$.

Let $T$ be a tree, and let $n$ denote $|V(T)|$.
We construct a graph $G$ from $T$ by adding a vertex~$w$ and making it adjacent to all vertices of $T$.
Clearly, $G$ has feedback vertex number~$1$ since $G[V(G) \setminus \{w\}]$ is a tree.

\begin{lemma}\label{lem:okp:app-hard:forward}
    If $G$ is outer $k$-planar, then $\bw(T) \le k + 1$.
\end{lemma}

\begin{proof}
    Let $D = (w, v_1, \dots, v_n)$ be an outer $k$-planar drawing of $G$.
    We define $\sigma\colon v_i \mapsto i$ and show that $\sigma$ is a linear order of $V(T)$ of bandwidth at most~$k+1$.
    Observe that if $G$ contains the edge $e = \{v_i,v_j\}$ with $i < j$, then $e$ crosses each edge $\{w, v_\ell\}$ with $i < \ell < j$.
    This implies that $j - i - 1 \le k$.
    Hence, the stretch of $e$ is at most $k+1$.
\end{proof}

\begin{lemma}\label{lem:okp:app-hard:backward}
    Let $b \ge 1$.  If $\bw(T) \le b$, then $G$ is outer $(5b-5)$-planar.
\end{lemma}

\begin{proof}
    Let $\sigma$ be a linear order of $V(T)$ with bandwidth at most~$b$.
    Assume that $\sigma(v_i) = i$, that is, $\sigma$ is specified by the sequence $(v_1, \dots, v_n)$.
    Let $D = (w, v_1, \dots, v_n)$ be a drawing of~$G$.  We show that $D$ is $(5b-5)$-planar.
    To this end, we analyze the two categories of edges of~$G$ separately:
    the \emph{tree edges} (that is, the edges in~$T$) and the \emph{non-tree edges} (that is, the edges incident to~$w$).

    Let $e = \{w, v_i\}$ be a non-tree edge.
    Observe that this edge $e$ crosses only tree edges in~$D$.
    Suppose that $e$ crosses a tree edge~$f = \{v_j, v_{j'}\}$ in~$D$, where $j < j'$.
    Then, by definition $j < i < j'$ holds.
    As the stretch of~$f$ is at most~$b$, we have $i - j \le b - 1$ and $j' - i \le b - 1$;
    see \cref{fig:edge-stretch}.
    Now, consider the vertex set $S_i = \{v_{i'} : |i - i'| \le b - 1\}$.
    Since $T$ is a tree, $T[S_i]$ is a forest, and hence $T[S_i]$ has at most $|S_i| - 1 \le 2b-2$ edges.
    Thus, $e$ crosses at most $2b-2$ edges in~$D$.

    \begin{figure}[hb]
        \centering
        \includegraphics[width=0.5\linewidth]{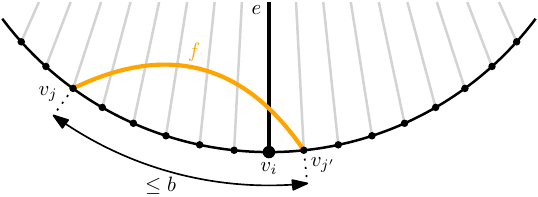}
        \caption{A partial view of drawing~$D$ with a tree edge $f$ intersecting a non-tree edge $e = \{w, v_i\}$.}
        \label{fig:edge-stretch}
    \end{figure}

    Now let $f = \{v_i, v_j\}$ be a tree edge with $i < j$.
    Then~$f$ crosses exactly $j - i - 1 \le b - 1$ non-tree edges.
    Moreover, $f$ crosses a tree edge $f' = \{v_{i'}, v_{j'}\}$ if and only if $i' < i < j'$ or $i' < j < j'$.
    Similarly to the previous discussion, at most $2b-2$ tree edges can satisfy the first condition and this also holds for the second.
    Hence, $f$ crosses at most $4b-4$ tree edges and at most $5b-5$ edges in total.
\end{proof}

Suppose that there is a polynomial-time $c$-approximation algorithm $\mathcal A$ for \OuterkPlanarity.  Let $T$ be a tree, and let $b=\bw(T)$.
By~\cref{lem:okp:app-hard:backward}, %
$\mathcal A$ would output an outer $5bc-5c$-planar drawing $D$ of $G$.
By~\cref{lem:okp:app-hard:forward}, $D$ can be transformed into a linear order of $V(T)$ with bandwidth at most $5bc$.
Thus, in polynomial time we can conclude that there exists a $5c$-approximate solution for \Bandwidth, which is impossible under $\cP \neq \NP$.
This completes the proof of \cref{thm:outer-k-planarity-np-hard}.

\subsection{XNLP-Hardness}
\label{sub:OkP:XNLP-hardness}

In the proof of \cref{thm:outer-k-planarity-np-hard}, we reduced the
gap-version of \Bandwidth to \OuterkPlanarity.
We exploited the gap to accommodate the crossings between edges in the original instance, which may increase the crossing number of each edge by $O(b)$.
However, if we allow parallel edges, we can reduce from (the exact version of) \Bandwidth by making the edges incident to $w$ sufficiently thick that we can ignore the $O(b)$ increase in the crossing numbers.
The following theorem is shown by emulating those parallel edges with rigid structures.

\begin{theorem}\label{thm:outer-k-planar-w-hardness}
    \OuterkPlanarity is \XNLP-hard when parameterized by $k$.
\end{theorem}

\begin{proof}
    \newcommand{\CliquePath}{\mathrm{CP}}
    As in the proof of \cref{thm:two-sided-k-planar-w-hardness}, we show the claim by reducing from \Bandwidth.

    First, we define a gadget called a \emph{clique path}, denoted $\CliquePath(t, \ell)$, for every integer $t > 1$ and every odd number $\ell > 1$.
    Let $H_1, H_2, \dots, H_{\ell-1}$ be cliques of $t$ vertices and, for $i \in [\ell - 1]$, let $v_{i, 1}, v_{i, 2}, \dots, v_{i, t}$ be the vertices of~$H_i$.
    Then, $\CliquePath(t, \ell)$ is obtained by identifying $v_{i,t}$ and $v_{i+1, 1}$ for each $i \in [\ell - 2]$; see \cref{fig:outer-k-planarity-w-hardness-clique-path}.
    We call the $\ell$ vertices $v_{1, 1}, v_{2, 1}, \dots, v_{\ell-1, 1}, v_{\ell-1, t}$ \emph{anchor points}.
    As $\ell$ is odd, $(\ell-1)/2$ is an integer, and the vertex $v_{(\ell-1)/2, t}$ (and hence $v_{(\ell+1)/2,1}$) separates $\CliquePath(t, \ell)$ evenly:
    Each connected component after removing $v_{(\ell-1)/2, t}$ has exactly $(\ell-1)/2$ anchor points.
    We refer to this vertex as the \emph{middle vertex} of $\CliquePath(t, \ell)$.
    By appropriately choosing $t$, $\ell$, and $k$, the vertices of $\CliquePath(t, \ell)$ appear consecutively in any outer $k$-planar drawing.
    Intuitively, the clique path behaves as a single vertex, and its anchor points emulate the end vertices of $\ell$ parallel edges.

    \begin{figure}[h]
        \centering
        \includegraphics[page=1]{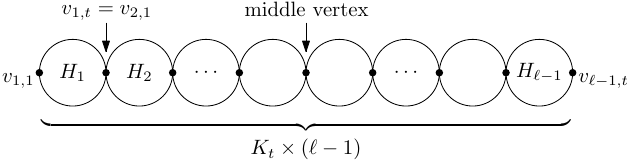}
        \caption{The clique path $\CliquePath(t, \ell)$ with $\ell$ anchor points.}
        \label{fig:outer-k-planarity-w-hardness-clique-path}
    \end{figure}

    Now we construct an instance $(G, k)$ of \OuterkPlanarity from an instance $(T, b)$ of \Bandwidth.
    Without loss of generality, we assume that $b \geq 3$.
    We let $t = 4(b^2+1) + 2$, $\ell = 4b^3 + 1$, and $k = ((t-2)/2)^2 = 4(b^2+1)^2 = 4b^4 + 8b^2 + 4$.
    Let us note that a clique of $t$ vertices admits an outer $k$-planar drawing.
    Moreover, each vertex in the clique is incident to an edge that has exactly $k$ crossings in any outer $k$-planar drawing.
    Starting with an empty graph $G$, we add a clique path $\CliquePath(t, \ell)$ for each vertex $v \in V(T)$ and denote it by $\CliquePath_v$.
    Then, we add a vertex $w$ with edges connecting to all anchor points in $G$.
    Lastly, for each edge $\{u, v\} \in E(T)$, we add an edge between the middle vertices of $\CliquePath_u$ and $\CliquePath_v$.

    Suppose that there is a linear order $(v_1, v_2, \dots, v_{n})$ of $V(T)$ with bandwidth at most~$b$.
    Then, we construct a cyclic order of $V(G)$ by aligning $(w, \CliquePath_{v_1}, \CliquePath_{v_2}, \dots, \CliquePath_{v_{n}})$ in this order, where the inner order of the vertices of each clique path $\CliquePath_{v_i}$ is
    \begin{align}\label{eq:clique-path}
        (v_{1, 1}, v_{1, 2}, \dots, v_{1, t}, v_{2, 1}, \dots, v_{2,t}, \dots, v_{\ell-1, 1}, \dots, v_{\ell-1, t}).
    \end{align}
    Observe that the edges of a clique in a clique path only cross the edges in the same clique, and their crossing numbers are at most $((t-2)/2)^2 = k$.
    Thus, in the rest of the proof, we can ignore the crossings involved in the edges of the cliques.
    As in the proof of \cref{lem:okp:app-hard:backward}, $G$ has two types of edges: The edges incident to $w$ and the edges between two middle vertices, which corresponds to edges of $T$.
    Following the same analysis as in \cref{lem:okp:app-hard:backward}, each edge of the first type crosses at most $2b - 2$ edges in $D$.
    Let $e$ be an edge of the second type that connects the middle vertices of $\CliquePath_{u}$ and $\CliquePath_{v}$ for some $u, v \in V(T)$.
    This edge $e$ crosses at most $4b-4$ edges of the second type in $D$.
    Moreover, there are at most $2(\ell-1)/2 + \ell(b-1) = \ell b - 1$ anchor points between the middle vertices in $D$, each of which has an incident edge of the first type that crosses $e$.
    Hence, $e$ crosses at most $\ell b + 4b - 5 = 4b^4 + 5b - 5 < k$ edges in total.
    Therefore, $D$ is an outer $k$-planar drawing of $G$.

    To show the other direction, we first observe that the vertices in each clique path $\CliquePath_{v}$ appear consecutively as (\ref{eq:clique-path}) in any outer $k$-planar drawing of $G$.
    We say that two circular drawings $D = (v_1, \dots, v_h)$ and $D' = (v'_1, \dots, v'_h)$ of a graph are \emph{isomorphic} if the mapping $v_i \mapsto v'_i$ is an automorphism of the graph.

    \begin{claim}\label{claim:outer-k-planarity-w-hardness-clique-path-consecutive}
        Let $D$ be an outer $k$-planar drawing of $G$.
        Then, for each $v \in V(T)$, the vertices in the clique path $\CliquePath_{v}$ appear consecutively as (\ref{eq:clique-path}) in $D$, which is unique up to isomorphism.
    \end{claim}
    \begin{claimproof}
        Let $D^* = (w, u_1, \dots, u_h)$ be the subdrawing of $D$ induced by $w$ and the vertices in~$\CliquePath_v$.
        In the following, $D$ and $D^*$ are considered to be linear orders starting from $w$.
        For each $u_j$, there is an incident edge that crosses exactly $k$ edges of the same clique $H_i$ in $D^*$.
        We call such edges \emph{critical edges} in $H_i$.
        Observe that if an edge $e = \{u_p, u_q\}$ with $p < q$ is a critical edges in $H_i$, there are exactly $(t-2)/2$ vertices in $H_i$ between $u_p$ and $u_q$, not including $u_p$ or $u_q$.
        This also implies that $e$ is a unique critical edge incident to $u_p$ and $u_q$.

        We now claim that the vertices in a clique $H_i$ of $\CliquePath_v$ appear consecutively in $D$.
        Suppose otherwise.
        Then there are two vertices $u_p$ and $u_q$ with $p < q$ that belong to the same clique~$H_i$ such that at least one vertex $w' \notin V(H_i)$ appears between them in $D$.
        We choose the largest $p$ and the smallest $q$ satisfying the above condition.
        Note that the vertex $w'$ may not belong to $\CliquePath_v$.
        As $w' \notin V(H_i)$, there is a path between $w$ and $w'$ that avoids vertices in $H_i$.
        Let $e$ and $e'$ be the (possibly identical) critical edges of $H_i$ incident to $u_p$ and $u_q$, respectively.
        If one of the critical edges $e$ and $e'$ jumps over $w'$, at least one edge of $P$ crosses this critical edge, which violates the $k$-planarity of $D$.
        Thus, the other end of $e$ appears before $w'$ in $D$ and the other end of $e'$ appears after $w'$ in $D$, that is, it holds that $e = \{u_p, u_{p'}\}$ for some $p' < p$ and $e' = \{u_q, u_{q'}\}$ for some $q < q'$.
        Since both $e$ and $e'$ are critical edges in $H_i$, there are exactly $(t-2)/2$ vertices between $u_p$ and $u_{p'}$ and exactly $(t-2)/2$ vertices between $u_{q}$ and $u_{q'}$, which are disjoint.
        This contradicts the fact that $H_i$ has $t$ vertices.

        We next claim that the vertices in $\CliquePath_v$ appear consecutively as (\ref{eq:clique-path}) in $D$.
        Suppose otherwise.
        Since all the vertices in $H_i$ are consecutive in $D$ for all $i$, they must be ordered as (\ref{eq:clique-path}) except for two extreme anchor points $v_{1, 1}$ and $v_{\ell-1, t}$.
        Suppose that $u_1 \neq v_{1,1}$.
        Let $e$ be the critical edge incident to $u_1$.
        Since $e$ cannot cross the edge $\{w, v_{1,1}\}$ in $D$ due to its criticality, the other end of $e$ appears before $v_{1,1}$.
        By considering the critical edge $e'$ incident to $v_{1,t}$, we can derive a contradiction similar to the one above.
    \end{claimproof}

    Now suppose that $G$ has an outer $k$-planar drawing $D = (w, w_1, \dots, w_N)$, where $N = |V(G)|$.
    By~\cref{claim:outer-k-planarity-w-hardness-clique-path-consecutive}, the vertices in each clique path appear consecutively in $D$ for each $1 \le i \le n$, that is, $D$ is of the form $(w, \CliquePath_{v_1}, \CliquePath_{v_2}, \dots, \CliquePath_{v_n})$.
    Let $\sigma$ be the linear order on $V(T)$ defined as $\sigma(v_i) = i$ for $v_i \in V(T)$.
    Consider an edge $\{v_i, v_j\} \in E(T)$ with $i < j$.
    As we discussed in the proof of \cref{thm:outer-k-planarity-np-hard}, the edge between the middle vertices of $\CliquePath_{v_i}$ and $\CliquePath_{v_j}$ must cross the edges that connect $w$ and anchor points between the middle vertices.
    The number of such anchor points is at most $k$ due to the $k$-planarity of $D$.
    Since there are at least $\ell(j - i - 1) + \ell - 1$ anchor points between them, we have $\ell(j - i) -1 \le k$.
    Therefore, it holds that $j - i \leq (k+1)/\ell = b + (8b^2-b+5)/(4b^3+1)$, which is strictly
    less than $b + 1$ if $b \geq 3$.
\end{proof}

\section{Open Problems}
\label{sec:open}

We conclude with a number of problems that we have left open in this paper.
\begin{itemize}
    \item We conjecture that \OuterkPlanarity is \XALP-complete (see \cite{BodlaenderGJPP22:IPEC:XALP} for the definition).
    \item Can we extend the algorithm for \TwoSidedkPlanarity to obtain an
          \XP-algorithm for $h$-layer $k$-planarity parameterized
          by~$h + k$?
    \item Another way to extend $k$-planarity is to consider
          \emph{min-$k$-planarity}, which is also called \emph{weak
              $k$-planarity} \cite{BinucciBBDD00LM24:min-k-planar:JGAA,
              campbell-etal-product-structure-CPC24}.  In a min-$k$-planar
          drawing, in every crossing, at least one of the two edges must
          have at most $k$ crossings. Are there \XP-algorithms for recognizing 2-layer min-$k$-planar graphs and outer min-$k$-planar graphs with respect to $k$?
    \item \TwoSidedkPlanarity can be seen as a restricted version of
          \OuterkPlanarity for bipartite graphs where the vertices of the
          two sets of the bipartition must not interleave in the cyclic
          vertex order.  This can be generalized as follows: For $k \ge 3$,
          can we efficiently recognize $k$-partite graphs that admit a
          $k$-planar straight-line drawing on the regular $k$-gon?  A
          related question has been investigated for fixed-order book
          embedding \cite{ackssuw-ecog-SWAT24}.
\end{itemize}

\paragraph*{AI Usage Disclosure}

Generative AI models by OpenAI were used to check and correct the proofs and to polish the text both originally written by the authors.
Especially, GPT-5.6 Sol Pro spotted a flaw in the dynamic programming for \OneSidedkPlanarity (and hence, for \TwoSidedkPlanarity) in the previous versions.

\bibliographystyle{plainurl}
\bibliography{main}

\end{document}